%% file: hal.tex
\documentclass[]{article}

\usepackage{tikz}
\usepackage{amsmath,amsthm,amssymb}
\usepackage{enumerate}
\usetikzlibrary{decorations.pathmorphing,decorations.pathreplacing,arrows,automata}

\theoremstyle{plain}\newtheorem{theorem}{Theorem}
\newtheorem{definition}[theorem]{Definition}
\newtheorem{lemma}[theorem]{Lemma}
\newtheorem{corollary}[theorem]{Corollary}

\newcommand{\ZZ}{\mathbb{Z}}

\newcommand{\NN}{\mathbb{N}}

\newcommand{\nspace}[1]{\textrm{\bf NSPACE}(#1)}
\newcommand{\ntime}[1]{\textrm{\bf NTIME}(#1)}
\newcommand{\horiz}{\mathcal H}

\newcommand{\totper}{\mathcal T}

\newcommand{\cne}{\textrm{\textbf{NE}}}
\newcommand{\ccone}{\textrm{\textbf{coNE}}}

\newcommand{\cnp}{\textrm{\textbf{NP}}}
\newcommand{\binaire}[1]{bin(#1)}
\newcommand{\unaire}[1]{un(#1)}
\newcommand{\Lh}[1]{\mathcal L^h_{#1}}
\newcommand{\Lt}[1]{\mathcal L_{#1}}

\begin{document}

\title{Periodicity in tilings}

\author{Emmanuel Jeandel \and Pascal Vanier}





\maketitle
\begin{abstract}
	Tilings and tiling systems are an abstract concept that arise both as a
	computational model and as a dynamical system.
	In this paper, we characterize the sets of periods that a tiling system can produce.
	We prove that up to a slight recoding, they correspond exactly to
	languages in the complexity classes $\nspace{n}$ and $\cne$.

\end{abstract}

\section*{Introduction}\label{S:one}
The model of tilings was introduced by Wang \cite{WangII} in the 60s to study
decison problems on some classes of logical formulas.
Roughly speaking, we are given some local constraints (a \emph{tiling system}),
and we consider colorings of the plane that respect these constraints (a \emph{tiling}).

This model has been discovered to be quite interesting both as a 
computation model \cite{Berger2} and as a two-dimensional analog 
of symbolic dynamics \cite{LindMarcus, Lind2}. In a sense, the results we
obtained are an answer to a question of symbolic dynamics via a computational point-of-view.

In this paper we consider \emph{periods} of tilings. A coloring of the
plane is said to be periodic of period $p$ if it is invariant by a
translation by a (horizontal/vertical) distance of $p$. 
It is easy to build a tiling system so that there exists a tiling of period $1$
but no tiling of any other period. It is also routine to build a tiling system for
which there is no tiling, hence no periodic tiling. The first important
result on the theory of tilings \cite{Berger2} 
states the existence of a tiling system which admits tilings, none of them being periodic.
Our main result in this article is a characterization of the sets of periods
we can obtain with tilings.

Periods in tilings are fundamental in the context of symbolic dynamics \cite{LindMarcus}.
One of the main purposes of multidimensional symbolic dynamics is to decide whether two
tiling systems generate rougly the same tilings or, more accurately, if the two
\emph{subshifts} they generate are conjugated. One of the main tools is
conjugacy invariants, i.e. quantities that are preserved under conjugacy.
The set of periods is such an invariant. One other well known example is the
\emph{entropy} which is a measure of the growth of the number of $n \times
n$ squares that can appears inside a tiling.

It turns out that a characterization of entropy in tilings, itself a
combinatorial quantity, is intimately linked with computational theory, more
precisely recursivity: entropy of tilings are exactly the (non-negative) right recursively
enumerable numbers \cite{HochMey}. This is merely an exception: most
conjugacy invariants are to be expressed in terms of recursivity
theory~\cite{AubrunS09,Meyero}.

We prove a similar theorem for periods in tilings, answering a open problem
in~\cite{Meyero}. For periodic tilings, which are truly \emph{finite
objects}, the good analog will not be \emph{recursivity} but
rather \emph{complexity}. We will see indeed in this article that sets of
tilings correspond to complexity classes, in particular to non-deterministic
exponential space and non-deterministic exponential time. These classes
are not unusual for a specialist in descriptive complexity \cite{EF:finmt}.
This is not a coincidence; the reader familiar with this topic may indeed
see the results in this article as an analog for tilings of the classical
result of Jones and Selman \cite{JonesSelman} on spectra of first order
formulas.

\section{Tilings, periodicity and computations}\label{S:Tpc}

Usually when considering tiling systems, Wang rules are used. Wang rules consider adjacent tiles only, whereas our set
of rules will consider a neighborhood of tiles.
This makes 
constructions easier to explain, while keeping almost the same properties as classical Wang 
tiling systems.

For any dimension $d\geq 1$, a tiling of $\ZZ^d$ with a finite set of tiles $T$ is a mapping $c:\ZZ^d\to T$. 
A tiling system is the pair $(T,I)$, where $I$ is a finite set of forbidden 
patterns $I\subset T^N$, where $N\subset_{finite} \ZZ^d$ is the neighborhood.
A tiling $c$ is said \emph{valid} if and only if none of the patterns of $I$ ever appear in $c$. 
Since the number of forbidden patterns is finite, we could specify the rules by \emph{allowed}
patterns as well. We give an example of such a tiling system with the tiles of figure~\ref{ex1_tuiles}a 
and the forbidden patterns of figure~\ref{ex1_tuiles}b. The allowed tilings are showed in figure~\ref{ex1_pavages}.

\begin{figure}[htbp]
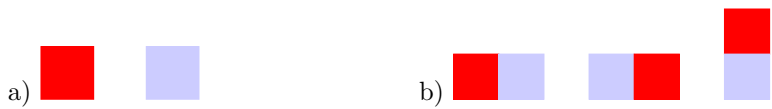

 \begin{center}
 \input figures/exemple_tuiles.tex
 \caption{The set of tiles (a) and the forbidden patterns (b).}\label{ex1_tuiles}
 \end{center}
\end{figure}

\begin{figure}[htbp]
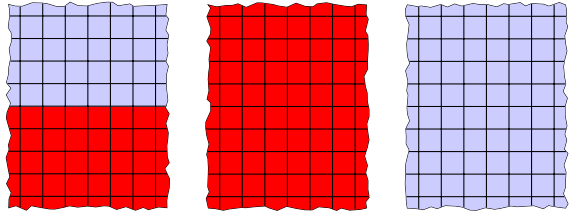

 \begin{center}
 \input figures/exemple_motifs.tex
 \caption{The only valid tilings of the system.}\label{ex1_pavages}
 \end{center}
\end{figure}

A tiling $c$ of dimension $d=2$ is said to be \emph{horizontaly periodic} if and only if there exists a 
\emph{period} $p\in\NN^*$ such that for all $x,y\in\ZZ$, $c(x,y)=c(x+p,y)$. A tiling $c$ of $\ZZ^d$ is 
\emph{periodic} if it has the same period on all its dimensions: 
$$c(x_1,x_2,\dots,x_d)  =  c(x_1+p,x_2,\dots,x_d) = c(x_1,x_2,\dots,x_d+p)$$

A tiling system is \emph{aperiodic} if and only if it tiles the plane but 
there is no valid periodic\footnote{On none of the dimensions.} tiling. Such tiling systems have been shown to exist with Wang rules, and by 
extension for our set of rules. J. Kari and P. Papasoglu gave such a tiling system with an 
interesting property: determinism.  A tiling system is \emph{NW-deterministic} if given two
diagonally adjacent tiles, there is only one way to complete the square in order to form a valid 
tiling. This means the two tiles \emph{force} the next one, the mechanism is shown in figure~\ref{deter}.

\begin{figure}[htbp]
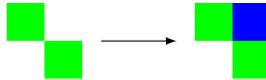

 \begin{center}
 \input figures/tuiles-deter.tex
 \caption{Determinism in the Kari-Papasoglu aperiodic tiling system.}\label{deter}
 \end{center}
\end{figure}

Since the tiling system was described with Wang rules, it is easy to transpose it to our set of rules and
change some details: instead of having the tiles to be adjacent diagonally, we can have them 
simply adjacent, as in figure~\ref{Edeter}. The type of determinism of this example will be called 
East-determinism. With such a tiling system, if a column of the plane is given, the half plane on its right is
then determined.

\begin{figure}[htbp]
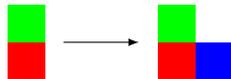

 \begin{center}
 \input figures/E-deter.tex
 \caption{East-determinism in a tiling system.}\label{Edeter}
 \end{center}
\end{figure}

\section{Turing machines}

As said earlier, tilings and recursivity are intimately linked.
In fact, it is quite easy to encode Turing machines in tilings. We give in
this section such an encoding. Similar encodings can be found e.g. in \cite{KariRevCA,Chaitin08}.

\begin{figure}[htbp]
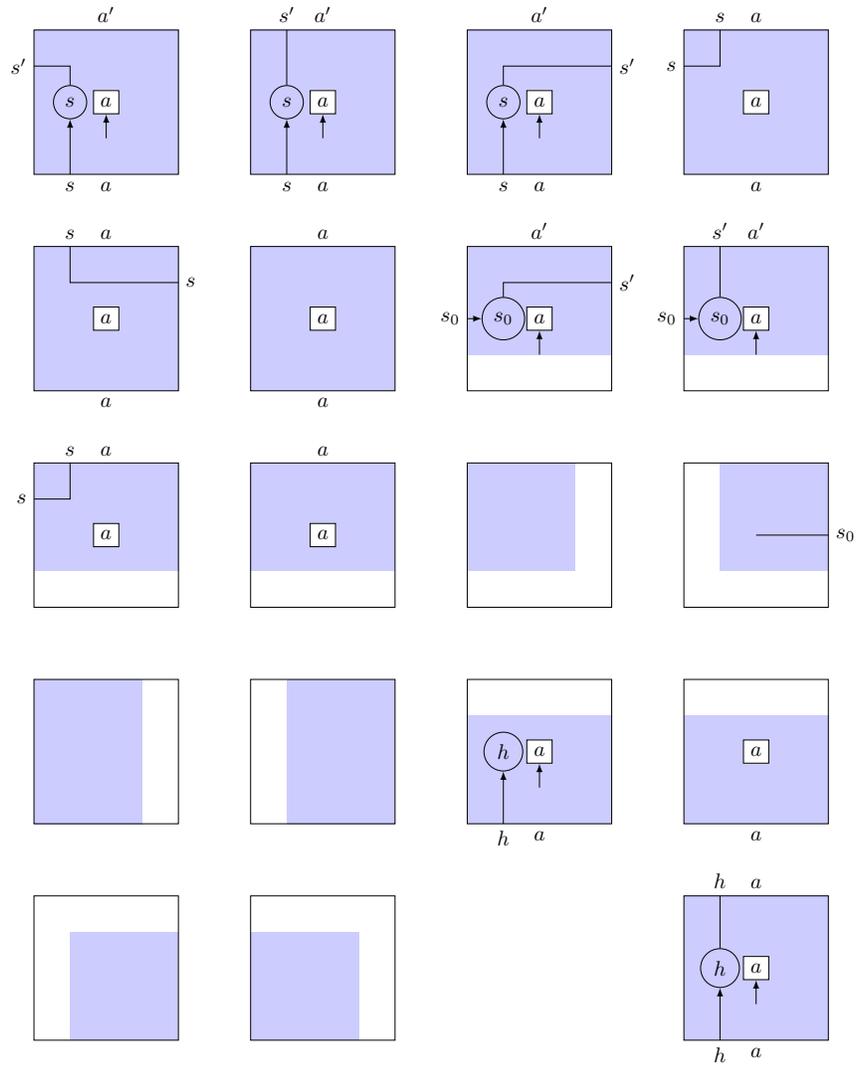

 \begin{center}
 \input figures/tuiles_MT.tex
 \caption{A tiling system, given by Wang tiles, simulating a Turing machine}\label{mt_tuiles}
 \end{center}
\end{figure}

Given a (non-deterministic, with one semi-infinite tape) Turing machine $M$, we build a tiling system $\tau_M$
in figure~\ref{mt_tuiles}. The tiling system is given by Wang tiles, i.e., we can
only glue two tiles together if they coincide on their common edge.
We now give some details on the picture:
\begin{itemize}
\item $s_0$ in the tiles is the initial state of the Turing machine.
\item The first tile corresponds to the case where the Turing machine, given the state $s$ and the letter $a$
chose to go to the left and to change from $s$ to $s'$, writing $a'$. The two
other tiles are similar.
\item $h$ represents a halting state. Note that the only states that can
  appear in the last step of a computation (before a border appears)
are halting states.
\end{itemize}

This tiling system $\tau_M$ has the following property: there is an accepting path
for the word $u$ in time (less than) $w$ using space (less than) $h$ if and only if we can tile a
rectangle of size $(w+2) \times h$  with white borders, the  first row containing the input.

Note that this method works for both deterministic and non-deterministic machines.
Using some usual tiling techniques, it is quite straightforward to build a
tiling system encoding computation paths of \emph{$n$-tape} Turing machine.
As a consequence, complexity results stated below are also valid when the
complexity classes are defined over $n$-tape Turing machines rather than
one-tape Turing Machines.

\section{Recognizing languages with tilings}

We introduce here a way to recognize languages over $\NN$, by looking at the set of periods of 
a given tiling system $\tau$.

We define $\Lh{\tau}$, the language recognized by a tiling system $\tau$, as the set of integers $p$ such
that there is a valid tiling of $\ZZ^2$ by $\tau$ of horizontal \emph{eigenperiod} $p$. 
We mean by \emph{eigenperiod}
that $p$ is the smallest non-zero period of the tiling, in order to avoid also having all multiples of 
any recognized number in the langage.
We similarly note $\Lt{\tau}$ the set of eigenperiods of $\tau$.

\begin{definition}
$\totper$ is the set of all languages $\Lt{\tau}$, where $\tau$ is any
tiling system.
$\horiz$ is the set of all languages $\Lh{\tau}$, where $\tau$ is any tiling
system.
$\totper$ and $\horiz$ are the classes of languages recognized by tiling and recognized horizontally 
by tiling respectively.
\end{definition}

We say that a language $\mathcal L$ is recognized by a tiling system $\tau$ if and only if
$\mathcal L=\Lh{\tau}$ or $\mathcal L=\Lt{\tau}$, depending on the context. 

A natural question that arises when studying languages concerns closure by union, intersection 
and complementation. Closure by union is easy to prove: take the disjoint union of the sets of tiles 
and forbid tiles from the first tiling system to be neighbors of tiles 
from the second one. 
Closure by intersection is not as straightforward, since the
classical construction by cartesian product of the two tiling systems would lead to a language 
containing the \emph{lcm} of the periods of each language. 

In this paper, we prove the following:
\begin{theorem} 
 $\horiz$ is closed by union, intersection and complementation.
\end{theorem}

\begin{theorem} \label{totperne}
$\totper$ is closed by intersection. $\totper$ is closed by complementation if and only if
\mbox{$\cne = \ccone$}.
\end{theorem}
Here $\cne$ is the class of languages recognized by a (one-tape) non-deterministic
Turing machine in time $2^{cn}$ for some $c > 0$. 
Note that for theorem
\ref{totperne} we need to work in any dimension $d$. That is, if
$\mathcal{L}_1, \mathcal{L}_2$ are the set of periods of the 
tiling system $\tau_1, \tau_2$, then there exists a
tiling system $\tau'$ that corresponds to the set of period $\mathcal{L}_1 \cap
 \mathcal{L}_2$. However even if $\tau_1$ and $\tau_2$ are two-dimensional
tiling systems, $\tau'$ might be in dimension greater than $2$.

To prove these theorems, we will actually give a characterization in terms
of structural complexity of our classes $\horiz$ and $\totper$.
This will be the purpose of the next two sections.
For an introduction to structural complexity, we suggest \cite{BDG,BDG2}.

\section{$\horiz$ and $\nspace{2^n}$}

To formulate our theorem, we consider sets of periods, i.e. subsets of
$\mathbb{N}^\star$, as unary or binary languages.
If $L \subset \mathbb{N}^\star$ then we define
$\unaire{L} = \{ 1^{n-1} | n \in L\}$. We define $\binaire{L}$ to be the
set of binary representations (missing the leading one) of numbers of $L$.
As an example, if $L = \{1, 4, 9 \}$, then
$\unaire{L} = \{ \epsilon, 111, 11111111\}$ and
$\binaire{L} = \{ \epsilon, 00, 001\}$.
Note that any language over the letter $1$ (resp. the letters $\{0,1\}$)
is the unary (resp. binary) representation of some subset of
$\mathbb{N}^\star$.

We now proceed to the statement of the theorem:
\begin{theorem}\label{horizspace}
Let $\mathcal L$ be a language, the following statements are equivalent:
\begin{enumerate}[i)]
 \item $\mathcal L\in\horiz$
 \item $\unaire{\mathcal L}\in \nspace{n}$
 \item $\binaire{\mathcal L}\in \nspace{2^n}$
\end{enumerate}
\end{theorem}
Recall that $\nspace{n}$ is the set of languages recognized by a
(one-tape) non-deterministic Turing machine in space $O(n)$.

The $(ii)\Leftrightarrow (iii)$ is folklore from computational complexity
theory. The following two 
lemmas will prove the equivalence $(ii) \Leftrightarrow (i)$ hence the result.

\begin{lemma}\label{lemmedirect}
For any tiling system $\tau$, $\unaire{\Lh{\tau}}\in \nspace{n}$.
\end{lemma}
\begin{proof}
Let $\Lh{\tau}$ be the language recognized by $\tau$.
We will construct a non-deterministic Turing machine accepting $n$ if and only if it is a 
horizontal eigenperiod of $\tau$. The 
machine has to work in space $\mathcal O(n)$, the input being given in unary.

\begin{figure}[htbp]
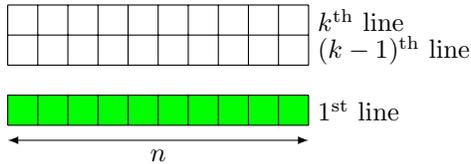

 \begin{center}
 \input figures/mt_nspacen.tex
 \caption{The tape of the machine when treating the $k^\textrm{th}$ line: we keep the firt line
in memory to verify the validity of the last one, and we guess the next line non-deterministically.}\label{mt_evol}
 \end{center}
\end{figure}

The tiling system has a finite number of finite rules, hence it has a "radius" of action $r$.
A consequence of this is that if there exists a horizontally periodic tiling 
of period $n$, then there is such a tiling of height at most $|T|^{rn}$.

This means that with a Turing machine working in space $n\times|T|^{rn}$, it is easy to verify if an integer $n$ 
is a period of the tiling $\tau$. 
We just guess non-deterministically a rectangle $n\times |T|^{rn} $ and check 
wether it is valid with copies of itself as neighbors. To check if it is an \emph{eigen}period we just have to keep in 
memory a boolean for each integer smaller than $n$, and check for each of them if they are not a possible 
period.

Actually, we only have to remember $r$ lines of the tiling at once to check the validity of some line. 
This allows us to bring the space back to $\mathcal O(n)$.
The method is the same as stated before, we check non-deterministically each line, all the preceeding
$r-1$ lines being stored to be able to check validity. We also keep the $r-1$ first lines guessed, in 
order to check the validity of the last lines. Figure~\ref{mt_evol} shows the state of the tape while 
checking the $k$-th line for $r=1$.

The machine hereby constructed works in space $\mathcal O(n)$, therefore 
$\Lh{\tau}\in\nspace{2^n}$.
\end{proof}

\begin{lemma}\label{lemmeconstr}
	
For any unary language $L \in \nspace{n}$, then
$\{n \in \mathbb{N}^\star, 1^{n-1} \in L \}\in\horiz$. 
\end{lemma}
\begin{proof}
Let $L \in \nspace{n}$.
There exists then a non-deterministic Turing machine $M$ accepting $L$ in linear space. 
Using traditional tricks from complexity theory, we can suppose that, 
on input $1^n$ the Turing machine uses exactly $n+1$ cells of the tape
(i.e. the input, with one additional cell on the right) and works in time
$c^n$ for some constant $c$.

We will build a tiling system $\tau$ so that $1^{n} \in L$
if and only if $n+4$ is a period of the tiling $\tau$.
The modification to obtain $n+1$ rather that $n+4$, and thus prove the
lemma, is left to the reader (basically ``fatten'' the gray tiles
presented below so that they absorb 3 adjacent tiles), and is of no interest rather than technical.

The tiling system will be made of several 
components, each of them having a specific goal. The components and their rules are 
as follows:

\begin{itemize}
 \item The first component $A$ is composed of an aperiodic E-deterministic tiling system,
 whose tiles will be called "whites". We take the one from section~\ref{S:Tpc}. We add a "gray"
 tile. The rules forbid any
 pattern containing a white tile above or below a gray tile. Hence a column containing a 
 gray tile can only have gray tiles.
 
 With this construction, an aperiodic tiling of period $p$ will have gray columns at distance
 $p$ of each other. The reason for that is that
 the white tile can only tile the plane aperiodically.
 
 For the moment nothing forbids more than one gray column to appear inside a period. Figure~\ref{compoA} shows
 a possible form of a periodic tiling at this stage.

\begin{figure}[htbp]
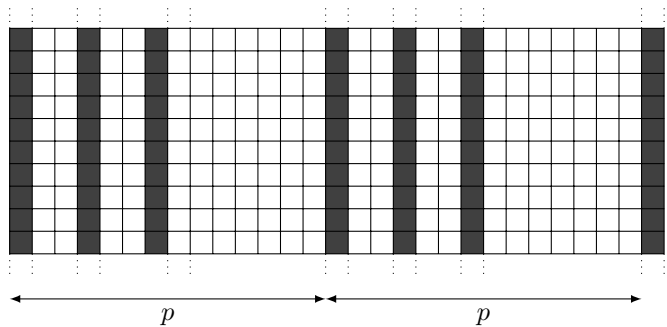

 \begin{center}
 \input figures/apres_compoA.tex
 \caption{A periodic tiling with the tiling system $A$.}\label{compoA}
 \end{center}
\end{figure}

 \item The second component $D=P\times\{R,B\}$ is a tile set that will allow us to force the
 distance between gray columns to be always the same. Furthermore, if we call $n$ this distance,
 then it will also force rows at distance $c^n$ of each other. It is divided into two subcomponents
 $P$ and $\{R,B\}$, which are respectively a counter in base $c$ and a marker for the rows. Let us
 focus on the counter: it adds one on each row, which means, if the number is delimited by two gray
 columns at distance $n$, it can go from $0$ to $c^n-1$. Such a counter, can be realized with a 
 transducer corresponding to the local operation.

\begin{figure}[htbp]
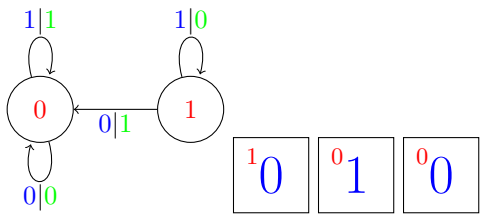

 \begin{center}
 \input figures/transducteur.tex
 \caption{Transducer corresponding to the operation of adding a bit and its 
 corresponding tiles. The states of the transducer correspond to the carry.}\label{transdu}
 \end{center}
\end{figure}

\begin{figure}[htbp]
 \begin{center}
 \input figures/exemple_transdu.tex
 \caption{Example of a valid tiling with the transducer tiles.}\label{transdu_pave}
 \end{center}
\end{figure}

 The transducer of figure~\ref{transdu} corresponds to the counter in the case $c=2$. 

 To make a set of tiles out of this, we only need tiles containing the carry and the value at the 
 same time. We forbid two tiles to be neighbors if the one on the left is not the result of 
 applying the transducer to the one on the right. The gray columns being the delimiters, if 
 a gray tile is on the right of some white tile, it is always considered as a tile containing a 
 carry. 

 The rules on the $\{R,B\}$ subcomponent are easy: a tile on the right or the left of a $R$ tile is 
 always a $R$ tile, the value inside $R$ tiles can only be $0$. A gray tile is $R$ if the tile below it
 was a $c-1$, that is to say when the number is reinitialized. 

 We made here rows of $R$ distant of $c^n$, which in conjunction with the gray columns form regular
 rectangles on the plane.
 
 Figure~\ref{compoDT}a shows some typical tiling at this stage: The aperiodic component can still be
 different between two gray columns, so the distance between these columns is not the period of the 
 tiling.

 \item Component $T$ is only a copy of $A$ which allows us to synchronize the first white columns 
 after the gray columns: synchronising these columns ensures that the aperiodic components between two
 gray columns are always the same, since the aperiodic tiles are E-deterministic. The rules are simple, 
 two horizontal neighbors have the same value on this component and a tile having a gray tile on its
 left has the same value in $A$ as in $T$.
 
 At this stage, we have regular rectangles on all the plane, whose width correspond to the period of 
 the tiling, as shown in figure~\ref{compoDT}b.
\begin{figure}[tbp]
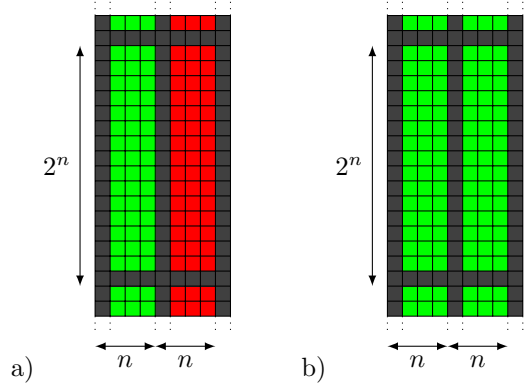

 \begin{center}
 \input figures/apres_compoD
 \caption{a) The form after adding the component $D$, the aperiodic components between
 two gray columns can be different. b) After adding also the component $T$, the aperiodic components 
 are exactly the same.}\label{compoDT}
 \end{center}
\end{figure}
 
 \item The last component $M$ is the component allowing us to encode Turing machines in each rectangle.
 We use the encoding $\tau_M$ we described previously in section~\ref{S:Tpc}.
 We force the computation to appear inside the white tiles : 
 The white bottom borders must appear only in the row $R$, and the
 row below will have top borders. And the two tiles between the row $R$
 and the gray tiles are corner tiles.
 Finally, the input of the Turing machine
 (hence the row above the row $R$) consist of only  ``1'' symbols, with a
 final blank symbol.

 The Turing machines considered here being non-deterministic, there could be different valid transitions
 on two horizontally adjacent rectangles, that is why we synchronize the transitions on each row.
 The method for the synchronization of the transition is almost the same as the method for
 the synchronisation of the aperiodic components, and thus not provided here.

\end{itemize}

Now we prove that
$1^{n} \in L$
if and only if $n+4$ is a period of the tiling system $\tau$.
First suppose that $n+4$ is a period and consider a tiling of period $n+4$.
\begin{itemize}
\item Due to component $A$, a gray column must appear. The period is a
succession of either gray and white columns.
\item  Due to the component $D$, the gray columns are spaced by a period of $p$,
$p < n+4$.
\item  Due to component $T$ the tiling we obtain is (horizontally) $p$-periodic when restricted to
the components $D,T,A$.
\item For the $M$ component to be correctly tiled, the input $1^{p-4}$
(4 = 1 (gray) + 1 (left border) + 1 (right border) + 1 (blank marker))
must be accepted by the Turing Machine, hence $1^{p-4} \in L$
\item  Finally, due to the synchronization of the non-deterministic transition, the
$M$ component is also $p$-periodic. As a consequence, our tiling is
$p$-periodic, hence $n+4 = p$. Therefore $1^n \in L$
\end{itemize}

Conversely, suppose $1^{n} \in L$.
Consider the coloring of period $n+4$ obtained as follows (only a period is
described):
\begin{itemize}
	\item The component $A$ consist of $n+3$ correctly tiled columns of our
	  aperiodic $E$-deterministic tiling systems, with an additional gray column.
	  As the $E$-deterministic tiling system tiles the plane, such a tiling
	  is possible
	\item The component $M$ corresponds to a successful computation path of
	  the Turing machine on the input $1^n$, that exists by hypothesis.	  
	\item We then add all other layers according to the rules to obtain a
	  valid configuration, thus obtaining a valid tiling of period exactly $n+4$.
  \end{itemize}	
\end{proof}

\begin{corollary}\label{corhorizspace}
The languages recognized horizontally by tiling are closed under intersection and complementation.
\end{corollary}
\begin{proof}
Immerman-Szelepczenyi's theorem \cite{Immerman88,BDG2} states that non-deterministic space complexity
classes are closed by complementation. The result is then a consequence of theorem~\ref{horizspace}.
\end{proof}
Note that if $\horiz$ is closed under
complementation, then $\nspace{2^n}$ is closed under complementation by theorem \ref{horizspace}. As a
consequence, any proof of corollary~\label{corhorizspace} will translate
into a proof that space complexity classes (greater than $2^n$) are closed
under complementation, hence give an alternate proof of
Immerman-Szelepczenyi's result.

This theorem could be generalized to tilings of dimension $d$ by considering tilings having a period $n$
on $d-1$ dimensions, as explained in \cite{Borchert}.


\section{$\totper$ and $\cne$}
We now proceed to total periods rather than horizontal periods.
We will prove:
\begin{theorem}\label{totpercne}
Let $\mathcal L \subset \mathbb{N}^\star$ be a language, the following statements are equivalent:
\begin{enumerate}[i)]
 \item there exists a dimension $d$ for which $\mathcal L\in\totper$
 \item $\unaire{\mathcal L}\in\cnp$
 \item $\binaire{\mathcal L}\in\cne$
\end{enumerate}
\end{theorem}
We will obtain as a corollary Theorem~\ref{totperne}. Note the slight
difference in formulation between Theorem~\ref{totperne} and 
Theorem \ref{horizspace}.
While we can encode a Turing machine working in time $n$ in a tiling of size $n^2$, we cannot check the validity of the tiling in less than $\mathcal O(n^2)$ time steps. And furthermore, we cannot, as far as we know,
check if $n$ is an eigenperiod in time less than $\mathcal O(n^3)$.
More generally, it is unclear that we
can encode a Turing machine working in time $n^c$ in a tiling of size less than $n^{2c}$. So instead
of considering the time class $\ntime{n}$, we consider the class $\cnp$ on unary inputs. This will
allow us to cover the gap.
The class $\cnp$ for unary inputs corresponds to the class $\cne$ for
binary inputs.

\begin{proof}

The statements $(i) \Rightarrow (ii) \Leftrightarrow (iii)$ were already explained. So we only 
have to prove $(ii) \Rightarrow (i)$.

In order to get a periodic structure from the construction of lemma~\ref{lemmeconstr}, we only need 
to slightly modify it. We will detail the modification for $d=2$,
the generalization being straightforward:
\begin{itemize} 
 \item in component A, we start from an aperiodic NW-deterministic tiling system. We 
   will consider 3 types of gray tiles: a crossing tile, a horizontal tile and a vertical tile. The rules are simple:
  \begin{itemize}
   \item above or below a vertical tile one can only have a crossing or a vertical tile, 
   \item on the left or on the right of a horizontal tile, there can only be a crossing or a horizontal tile,
   \item on the left and on the right of a crossing tile, there can only be a crossing or a horizontal tile,
   \item above or below a crossing tile, there can only be a crossing or a vertical tile.
  \end{itemize}
  With these rules, when a row and a column of gray tiles cross, the crossing can only be a crossing tile.
 \item In component $D$, instead of using a counter, we only use a signal propagating in diagonal from 
 each crossing tile, which can only end on a crossing tile. Instead of having rectangles of size $c^n\times n$
 we now have squares of size $n\times n$.
 At this step in the proof, the reasoning is a follows: Take a periodic
 tiling, then either a horizontal or a vertical tile must appear, as the
 aperioc tiling system cannot tile the plane periodically. Then by component $D$,
 an horizontal tile will force a vertical tile to appear, and vice versa, so
 that every periodic tilings is composed of squares.
 \item The component $T$ will synchronize the aperiodic component between
   all squares, by propagating the first line and first column of a square 
   to all the neighbors. Indeed, a square of a NW-deterministic tiling system is
   entirely fixed once we know the first line and the first column.
   This is easily done by 
 synchronising as shown in figure~\ref{transmiamelioree}.
 
\begin{figure}[htbp]
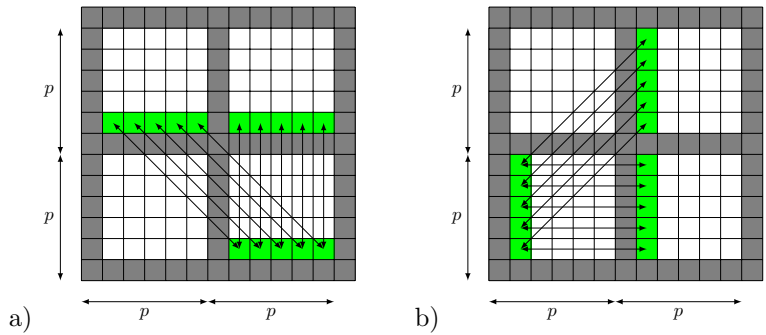

 \begin{center}
 \input figures/transmi.tex
 \caption{Transmission of the first row in a) and of the first line in b).}\label{transmiamelioree}
 \end{center}
\end{figure}

\item The Turing Machine component is the same as before.
\end{itemize}
Now this proves that every unary langage $L \in \ntime{n}$ is the set of
periods of a tiling system.

To prove the result for any unary language $L \in \cnp$, we will need to
work in higher dimensions : we will encode a machine working
in time $n^c$ in a tiling with $2c$ dimensions  ($c$ dimensions for the space and $c$ dimensions for the
time). The difficulty when generalising the dimension is not to keep the "hypercubical" structure,
as the previously described structure can straightforwardly be transposed in any dimension $d$,
but to keep the constraints of the Turing machine's tiles  local. That is to say, two tape positions $i$
and $i+1$ have to be neighbors. The same goes for time, a content of a certain tape position at time $t$ and its content at time $t+1$ have to be neighbors.

Therefore, we will bend time and space... Actually, we will \emph{fold} it as shown for dimension 3 in
figure~\ref{repliage}. 

\begin{figure}[htbp]
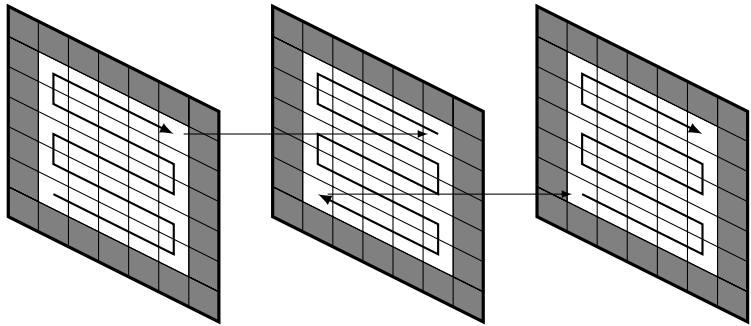

 \begin{center}
 \input figures/foldinglayers.tex
 \caption{Folding of the tape in dimension 3.}\label{repliage}
 \end{center}
\end{figure}

Such a folding has already been described by Borchert
\cite{Borchert} and can also be deduced from \cite{JonesSelman}.
We do not detail the construction which is technical and rather tedious.
The construction we obtain here allows us to code the Turing machine in a
hypercube of dimension $2c$, hence proving the result.
\end{proof}
\section*{Conclusion}
We obtain in this paper a complete characterization of the sets of periods
one can accomplish with tilings. The result is based on structural
complexity theory, and is very different from the sets of periods we can
obtain for $1$-dimensional tilings (subshifts of finite type, see
\cite{LindMarcus}) or for continuous discrete dynamical systems on the real
line (Sharkovskii theorem, see e.g. \cite{BGMS}).

The next step is to characterize the \emph{number} of tilings of period $p$
for a given tiling system $\tau$. Not surprisingly, preliminary work
suggests we can characterize these functions via the class $\# P$.

\bibliographystyle{plain}
\bibliography{biblio}
\end{document}

%% file: figures/exemple_tuiles.tex
a)~\begin{tikzpicture}[scale=0.7]
 \filldraw[color=red] (0,0) rectangle (1,1);
 \filldraw[color=blue!20] (2,0) rectangle (3,1);
 \filldraw[color=white] (4,0) rectangle (5,1);
\end{tikzpicture}
~~~~~~~~~~~~b)~\begin{tikzpicture}[scale=0.6]
 \filldraw[color=red] (0,0) rectangle (1,1);
 \filldraw[color=blue!20] (1,0) rectangle (2,1);
 \filldraw[color=blue!20] (3,0) rectangle (4,1);
 \filldraw[color=red] (4,0) rectangle (5,1);
 \filldraw[color=red] (6,1) rectangle (7,2);
 \filldraw[color=blue!20] (6,0) rectangle (7,1);
\end{tikzpicture}

%% file: figures/exemple_motifs.tex
\begin{tikzpicture}[scale=0.3]
 \clip[draw,decorate,decoration={random steps, segment length=3pt, amplitude=1pt}] (0.5,0.5) rectangle (7.5,9.5);
 \foreach \x in {0,...,7}
  \foreach \y in {0,...,4}
   \filldraw[color=red] (\x,\y) rectangle (\x+1,\y+1);
 \foreach \x in {0,...,7}
  \foreach \y in {5,...,9}
   \filldraw[color=blue!20] (\x,\y) rectangle (\x+1,\y+1);
 \draw (0,0) grid (8,10);
\end{tikzpicture}
~~
\begin{tikzpicture}[scale=0.3]
 \clip[draw,decorate,decoration={random steps, segment length=3pt, amplitude=1pt}] (10+0.5,0.5) rectangle (10+7.5,9.5);
 \foreach \x in {0,...,7}
  \foreach \y in {0,...,9}
   \filldraw[color=red] (10+\x,\y) rectangle (10+\x+1,\y+1);
 \draw (10+0,0) grid (10+8,10);
\end{tikzpicture}
~~
\begin{tikzpicture}[scale=0.3]
 \clip[draw,decorate,decoration={random steps, segment length=3pt, amplitude=1pt}] (10+0.5,0.5) rectangle (10+7.5,9.5);
 \foreach \x in {0,...,7}
  \foreach \y in {0,...,9}
   \filldraw[color=blue!20] (10+\x,\y) rectangle (10+\x+1,\y+1);
 \draw (10+0,0) grid (10+8,10);
\end{tikzpicture}

%% file: figures/tuiles-deter.tex
\begin{tikzpicture}[scale=0.5]
 \fill[color=green] (0,1) rectangle (1,2);
 \fill[color=green] (1,0) rectangle (2,1);
 \node[right] (l) at (2,1) {};
 \node[left] (r) at (5,1) {};
 \path[-latex] (l) edge (r);
 \fill[color=green] (5,1) rectangle (6,2);
 \fill[color=green] (6,0) rectangle (7,1);
 \fill[color=blue] (6,1) rectangle (7,2);
\end{tikzpicture}

%% file: figures/E-deter.tex
\begin{tikzpicture}[scale=0.5]
 \fill[color=red] (0,0) rectangle (1,1);
 \fill[color=green] (0,1) rectangle (1,2);
 \node[right] (l) at (1,1) {};
 \node[left] (r) at (4,1) {};
 \path[-latex] (l) edge (r);
 \fill[color=red] (4,0) rectangle (5,1);
 \fill[color=blue] (5,0) rectangle (6,1);
 \fill[color=green] (4,1) rectangle (5,2);
\end{tikzpicture}

%% file: figures/tuiles_MT.tex
\scalebox{0.8}{
\begin{tikzpicture}[auto,scale=0.6]
\filldraw[fill=blue!20] (0,0) rectangle (4,4);
\node[draw,circle] (etat) at (1,2) {$s$};
\draw[-latex] (1,0) node[below] (basetat) {$s$} -- (etat);
\draw (etat) -- (1,3) -- (0,3);
\node[draw,rectangle,fill=white] (ruban) at (2,2) {$a$};
\draw[-latex] (2,1) -- (ruban);
\node[left] at (0,3) (hautetat) {$s'$};
\node[above] at (2,4) (hautruban) {$a'$};
\node[below] at (2,0) (basruban) {$a$};
\begin{scope}[shift={(6,0)}]
\filldraw[fill=blue!20] (0,0) rectangle (4,4);
\node[draw,circle] (etat) at (1,2) {$s$};
\draw[-latex] (1,0) node[below] (basetat) {$s$} -- (etat);
\draw (etat) -- (1,4);
\node[draw,rectangle,fill=white] (ruban) at (2,2) {$a$};
\draw[-latex] (2,1) -- (ruban);
\node[above] at (1,4) (hautetat) {$s'$};
\node[above] at (2,4) (hautruban) {$a'$};
\node[below] at (2,0) (basruban) {$a$};
\end{scope}
\begin{scope}[shift={(12,0)}]
\filldraw[fill=blue!20] (0,0) rectangle (4,4);
\node[draw,circle] (etat) at (1,2) {$s$};
\draw[-latex] (1,0) node[below] (basetat) {$s$} -- (etat);
\draw (etat) -- (1,3) -- (4,3);
\node[draw,rectangle,fill=white] (ruban) at (2,2) {$a$};
\draw[-latex] (2,1) -- (ruban);
\node[right] at (4,3) (hautetat) {$s'$};
\node[above] at (2,4) (hautruban) {$a'$};
\node[below] at (2,0) (basruban) {$a$};
\end{scope}
\begin{scope}[shift={(18,0)}]
\filldraw[fill=blue!20] (0,0) rectangle (4,4);
\draw  (0,3) -- (1,3) -- (1,4);
\node[draw,rectangle,fill=white] (ruban) at (2,2) {$a$};
\node[left] at (0,3)  (basetat) {$s$};
\node[above] at (1,4) (hautetat) {$s$};
\node[above] at (2,4) (hautruban) {$a$};
\node[below] at (2,0) (basruban) {$a$};
\end{scope}
\begin{scope}[shift={(0,-6)}]
\filldraw[fill=blue!20] (0,0) rectangle (4,4);
\draw  (4,3) -- (1,3) -- (1,4);
\node[draw,rectangle,fill=white] (ruban) at (2,2) {$a$};
\node[right] at (4,3)  (basetat) {$s$};
\node[above] at (1,4) (hautetat) {$s$};
\node[above] at (2,4) (hautruban) {$a$};
\node[below] at (2,0) (basruban) {$a$};
\end{scope}
\begin{scope}[shift={(6,-6)}]
\filldraw[fill=blue!20] (0,0) rectangle (4,4);
\node[draw,rectangle,fill=white] (ruban) at (2,2) {$a$};
\node[above] at (2,4) (hautruban) {$a$};
\node[below] at (2,0) (basruban) {$a$};
\end{scope}
\begin{scope}[shift={(12,-6)}]
\fill[color=blue!20] (0,1) rectangle (4,4);
\draw (0,0) rectangle (4,4);
\node[draw,circle] (etat) at (1,2) {$s_0$};
\draw (etat) -- (1,3) -- (4,3);
\draw[-latex] (0,2) -- (etat);
\node[draw,rectangle,fill=white] (ruban) at (2,2) {$a$};
\draw[-latex] (2,1) -- (ruban);
\node[left] at (0,2) (basetat) {$s_0$};
\node[right] at (4,3) (hautetat) {$s'$};
\node[above] at (2,4) (hautruban) {$a'$};
\end{scope}
\begin{scope}[shift={(18,-6)}]
\fill[color=blue!20] (0,1) rectangle (4,4);
\draw (0,0) rectangle (4,4);
\node[draw,circle] (etat) at (1,2) {$s_0$};
\draw (etat) -- (1,4);
\draw[-latex] (0,2) -- (etat);
\node[draw,rectangle,fill=white] (ruban) at (2,2) {$a$};
\draw[-latex] (2,1) -- (ruban);
\node[above] at (1,4) (hautetat) {$s'$};
\node[above] at (2,4) (hautruban) {$a'$};
\node[left] at (0,2) (basetat) {$s_0$};
\end{scope}
\begin{scope}[shift={(0,-12)}]
\fill[color=blue!20] (0,1) rectangle (4,4);
\draw (0,0) rectangle (4,4);
\draw (0,3) -- (1,3) -- (1,4);
\node[draw,rectangle,fill=white] (ruban) at (2,2) {$a$};
\node[above] at (1,4) (hautetat) {$s$};
\node[left] at (0,3) (basetat) {$s$};
\node[above] at (2,4) (hautruban) {$a$};
\end{scope}
\begin{scope}[shift={(6,-12)}]
\fill[color=blue!20] (0,1) rectangle (4,4);
\draw (0,0) rectangle (4,4);
\node[draw,rectangle,fill=white] (ruban) at (2,2) {$a$};
\node[above] at (2,4) (hautruban) {$a$};
\end{scope}
\begin{scope}[shift={(12,-12)}]
\fill[color=blue!20] (0,1) rectangle (3,4);
\draw (0,0) rectangle (4,4);
\end{scope}
\begin{scope}[shift={(18,-12)}]
\fill[color=blue!20] (4,1) rectangle (1,4);
\draw (0,0) rectangle (4,4);
\draw (2,2) -- (4,2);
\node[right] at (4,2) (basetat) {$s_0$};
\end{scope}
\begin{scope}[shift={(0,-18)}]
\fill[color=blue!20] (0,0) rectangle (3,4);
\draw (0,0) rectangle (4,4);
\end{scope}
\begin{scope}[shift={(6,-18)}]
\fill[color=blue!20] (4,0) rectangle (1,4);
\draw (0,0) rectangle (4,4);
\end{scope}
\begin{scope}[shift={(12,-18)}]
\fill[color=blue!20] (0,3) rectangle (4,0);
\draw (0,0) rectangle (4,4);
\node[draw,circle] (etat) at (1,2) {$h$};
\draw[-latex] (1,0) node[below] (basetat) {$h$} -- (etat);
\node[draw,rectangle,fill=white] (ruban) at (2,2) {$a$};
\draw[-latex] (2,1) -- (ruban);
\node[below] at (2,0) {$a$};
\end{scope}
\begin{scope}[shift={(18,-18)}]
\fill[color=blue!20] (0,0) rectangle (4,3);
\draw (0,0) rectangle (4,4);
\node[draw,rectangle,fill=white] (ruban) at (2,2) {$a$};
\node[below] at (2,0) (basruban) {$a$};
\end{scope}
\begin{scope}[shift={(0,-24)}]
\fill[color=blue!20] (1,0) rectangle (4,3);
\draw (0,0) rectangle (4,4);
\end{scope}
\begin{scope}[shift={(6,-24)}]
\fill[color=blue!20] (3,0) rectangle (0,3);
\draw (0,0) rectangle (4,4);
\end{scope}
\begin{scope}[shift={(18,-24)}]
\filldraw[fill=blue!20] (0,0) rectangle (4,4);
\node[draw,circle] (etat) at (1,2) {$h$};
\draw[-latex] (1,0) node[below] (basetat) {$h$} -- (etat);
\draw (etat) -- (1,4);
\node[draw,rectangle,fill=white] (ruban) at (2,2) {$a$};
\draw[-latex] (2,1) -- (ruban);
\node[above] at (1,4) (hautetat) {$h$};
\node[above] at (2,4) (hautruban) {$a$};
\node[below] at (2,0) (basruban) {$a$};
\end{scope}
\end{tikzpicture}
}

%% file: figures/mt_nspacen.tex
\begin{tikzpicture}[scale=0.4]
 \fill[color=green] (0,0) rectangle (10,1);
 \draw (0,0) grid (10,1);
 
 \draw (0,2) grid (10,4);
 \node[right] at (10,3.5) {$k^\textrm{th}$ line};
 \node[right] at (10,2.5) {$(k-1)^\textrm{th}$ line};
 \node[right] at (10,0.5) {$1^\textrm{st}$ line};
 \draw[latex-latex] (10,-0.5) -- (0,-0.5) node[midway,below]{$n$};
\end{tikzpicture}

%% file: figures/apres_compoA.tex
\begin{tikzpicture}[scale=0.3]
 \fill[color=darkgray] (0,0) rectangle (1,10);
 \draw[dotted] (0,10) -- (0,11);
 \draw[dotted] (1,10) -- (1,11);
 \draw[dotted] (0,0) -- (0,-1);
 \draw[dotted] (1,0) -- (1,-1);
 \fill[color=darkgray] (3,0) rectangle (4,10);
 \draw[dotted] (3,10) -- (3,11);
 \draw[dotted] (4,10) -- (4,11);
 \draw[dotted] (3,0) -- (3,-1);
 \draw[dotted] (4,0) -- (4,-1);
 \fill[color=darkgray] (6,0) rectangle (7,10);
 \draw[dotted] (7,10) -- (7,11);
 \draw[dotted] (8,10) -- (8,11);
 \draw[dotted] (7,0) -- (7,-1);
 \draw[dotted] (8,0) -- (8,-1);
 \fill[color=darkgray] (14,0) rectangle (15,10);
 \draw[dotted] (14,10) -- (14,11);
 \draw[dotted] (15,10) -- (15,11);
 \draw[dotted] (14,0) -- (14,-1);
 \draw[dotted] (15,0) -- (15,-1);
 \fill[color=darkgray] (17,0) rectangle (18,10);
 \draw[dotted] (17,10) -- (17,11);
 \draw[dotted] (18,10) -- (18,11);
 \draw[dotted] (17,0) -- (17,-1);
 \draw[dotted] (18,0) -- (18,-1);
 \fill[color=darkgray] (20,0) rectangle (21,10);
 \draw[dotted] (20,10) -- (20,11);
 \draw[dotted] (21,10) -- (21,11);
 \draw[dotted] (20,0) -- (20,-1);
 \draw[dotted] (21,0) -- (21,-1);
 \fill[color=darkgray] (28,0) rectangle (29,10);
 \draw[dotted] (28,10) -- (28,11);
 \draw[dotted] (29,10) -- (29,11);
 \draw[dotted] (28,0) -- (28,-1);
 \draw[dotted] (29,0) -- (29,-1);
 \draw (0,0) grid (29,10);
 \draw[latex-latex] (0,-2) -- (14,-2) node[midway,below]{$p$};
 \draw[latex-latex] (14,-2) -- (28,-2) node[midway,below]{$p$};
\end{tikzpicture}

%% file: figures/transducteur.tex
\begin{tikzpicture}[auto,inner sep=1]
 \node[state] 	(q_0)		{$\color{red}0$};
 \node[state]	(q_1)   at (2,0)	{$\color{red}1$};
 \path[->]	(q_0)	edge [loop above]	node 	{$\color{blue}1\color{black}|\color{green}1$}	(q_0)
 		(q_0)	edge [loop below]	node	{$\color{blue}0\color{black}|\color{green}0$}	(q_0)
		(q_1)	edge 	node	{$\color{blue}0\color{black}|\color{green}1$}	(q_0)
		(q_1)	edge [loop above]	node	{$\color{blue}1\color{black}|\color{green}0$}	(q_1);
\end{tikzpicture}
\begin{tikzpicture}[scale=0.25]
 \draw (0,0) rectangle (4,4);
 \node (valeur) at (2,2) {\color{blue}\huge 0};
 \node (retenue) at (1,3) {\color{red}\small 1};
\end{tikzpicture}
\begin{tikzpicture}[scale=0.25]
 \draw (0,0) rectangle (4,4);
 \node (valeur) at (2,2) {\color{blue}\huge 1};
 \node (retenue) at (1,3) {\color{red}\small 0};
\end{tikzpicture}
\begin{tikzpicture}[scale=0.25]
 \draw (0,0) rectangle (4,4);
 \node (valeur) at (2,2) {\color{blue}\huge 0};
 \node (retenue) at (1,3) {\color{red}\small 0};
\end{tikzpicture}

%% file: figures/exemple_transdu.tex
\begin{tikzpicture}[scale=0.25]
 \draw (0,0) rectangle (4,4);
 \node (valeur) at (2,2) {\color{blue}\huge 0};
\begin{scope}[shift={(4,0)}]
 \draw (0,0) rectangle (4,4);
 \node (valeur) at (2,2) {\color{blue}\huge 1};
\end{scope}
\begin{scope}[shift={(8,0)}]
 \draw (0,0) rectangle (4,4);
 \node (valeur) at (2,2) {\color{blue}\huge 1};
\end{scope}
\begin{scope}[shift={(0,4)}]
 \draw (0,0) rectangle (4,4);
 \node (valeur) at (2,2) {\color{green}\huge 1};
 \node (retenue) at (1,3) {\color{red}\small 0};
\end{scope}
\begin{scope}[shift={(4,4)}]
 \draw (0,0) rectangle (4,4);
 \node (valeur) at (2,2) {\color{green}\huge 0};
 \node (retenue) at (1,3) {\color{red}\small 1};
\end{scope}
\begin{scope}[shift={(8,4)}]
 \draw (0,0) rectangle (4,4);
 \node (valeur) at (2,2) {\color{green}\huge 0};
 \node (retenue) at (1,3) {\color{red}\small 1};
\end{scope}
\begin{scope}[shift={(12,0)},scale=4]
\fill[color=darkgray] (0,0) rectangle (1,2);
\draw (0,0) grid (1,2);
\end{scope}
\end{tikzpicture}

%% file: figures/apres_compoD.tex
a)\begin{tikzpicture}[scale=0.2]
 \fill[color=darkgray] (0,0) rectangle (1,20);
 \draw[dotted] (0,20) -- (0,21);
 \draw[dotted] (1,20) -- (1,21);
 \draw[dotted] (0,0) -- (0,-1);
 \draw[dotted] (1,0) -- (1,-1);
 \fill[color=darkgray] (4,0) rectangle (5,20);
 \draw[dotted] (4,20) -- (4,21);
 \draw[dotted] (5,20) -- (5,21);
 \draw[dotted] (4,0) -- (4,-1);
 \draw[dotted] (5,0) -- (5,-1);
 \fill[color=darkgray] (8,0) rectangle (9,20);
 \draw[dotted] (8,20) -- (8,21);
 \draw[dotted] (9,20) -- (9,21);
 \draw[dotted] (8,0) -- (8,-1);
 \draw[dotted] (9,0) -- (9,-1);
 \fill[color=green] (1,0) rectangle (4,20);
 \fill[color=red] (5,0) rectangle (8,20);
 \fill[color=darkgray] (0,2) rectangle (9,3);
 \fill[color=darkgray] (0,18) rectangle (9,19);
 \draw (0,0) grid (9,20);
 \draw[latex-latex] (-1,2) -- (-1,18) node[midway,left]{$2^n$};
 \draw[latex-latex] (0,-2) -- (4,-2) node[midway,below]{$n$};
 \draw[latex-latex] (4,-2) -- (8,-2) node[midway,below]{$n$};
\end{tikzpicture}
~~~~~~
b)\begin{tikzpicture}[scale=0.2]
 \fill[color=darkgray] (0,0) rectangle (1,20);
 \draw[dotted] (0,20) -- (0,21);
 \draw[dotted] (1,20) -- (1,21);
 \draw[dotted] (0,0) -- (0,-1);
 \draw[dotted] (1,0) -- (1,-1);
 \fill[color=darkgray] (4,0) rectangle (5,20);
 \draw[dotted] (4,20) -- (4,21);
 \draw[dotted] (5,20) -- (5,21);
 \draw[dotted] (4,0) -- (4,-1);
 \draw[dotted] (5,0) -- (5,-1);
 \fill[color=darkgray] (8,0) rectangle (9,20);
 \draw[dotted] (8,20) -- (8,21);
 \draw[dotted] (9,20) -- (9,21);
 \draw[dotted] (8,0) -- (8,-1);
 \draw[dotted] (9,0) -- (9,-1);
 \fill[color=green] (1,0) rectangle (4,20);
 \fill[color=green] (5,0) rectangle (8,20);
 \fill[color=darkgray] (0,2) rectangle (9,3);
 \fill[color=darkgray] (0,18) rectangle (9,19);
 \draw (0,0) grid (9,20);
 \draw[latex-latex] (-1,2) -- (-1,18) node[midway,left]{$2^n$};
 \draw[latex-latex] (0,-2) -- (4,-2) node[midway,below]{$n$};
 \draw[latex-latex] (4,-2) -- (8,-2) node[midway,below]{$n$};
\end{tikzpicture}

%% file: figures/transmi.tex
a)\scalebox{0.7}{
\begin{tikzpicture}[scale=0.4]
 \fill[color=gray] (0,0) rectangle (1,12);
 \fill[color=gray] (6,0) rectangle (7,12);
 \fill[color=gray] (0,0) rectangle (12,1);
 \fill[color=gray] (0,6) rectangle (12,7);
 \fill[color=gray] (0,12) rectangle (13,13);
 \fill[color=gray] (12,0) rectangle (13,13);
 \fill[color=green] (7,1) rectangle (12,2);
 \fill[color=green] (7,7) rectangle (12,8);
 \fill[color=green] (1,7) rectangle (6,8);
 \draw (0,0) grid (13,13);
 \draw[latex-latex] (7.5,1.5) -- (1.5,7.5);
 \draw[latex-latex] (8.5,1.5) -- (2.5,7.5);
 \draw[latex-latex] (9.5,1.5) -- (3.5,7.5);
 \draw[latex-latex] (10.5,1.5) -- (4.5,7.5);
 \draw[latex-latex] (11.5,1.5) -- (5.5,7.5);
 
 \draw[latex-latex] (7.5,1.5) -- (7.5,7.5);
 \draw[latex-latex] (8.5,1.5) -- (8.5,7.5);
 \draw[latex-latex] (9.5,1.5) -- (9.5,7.5);
 \draw[latex-latex] (10.5,1.5) -- (10.5,7.5);
 \draw[latex-latex] (11.5,1.5) -- (11.5,7.5);
 \draw[latex-latex] (-1,0) -- (-1,6) node[midway,left]{$p$};
 \draw[latex-latex] (-1,6) -- (-1,12) node[midway,left]{$p$};
 \draw[latex-latex] (0,-1) -- (6,-1) node[midway,below]{$p$};
 \draw[latex-latex] (6,-1) -- (12,-1) node[midway,below]{$p$};
\end{tikzpicture}
}
~~~~
b)\scalebox{0.7}{
\begin{tikzpicture}[scale=0.4]
 \fill[color=gray] (0,0) rectangle (1,12);
 \fill[color=gray] (6,0) rectangle (7,12);
 \fill[color=gray] (0,0) rectangle (12,1);
 \fill[color=gray] (0,6) rectangle (12,7);
 \fill[color=gray] (0,12) rectangle (13,13);
 \fill[color=gray] (12,0) rectangle (13,13);
 \fill[color=green] (1,1) rectangle (2,6);
 \fill[color=green] (7,7) rectangle (8,12);
 \fill[color=green] (7,1) rectangle (8,6);
 \draw (0,0) grid (13,13);

 \draw[latex-latex] (1.5,1.5) -- (7.5,1.5);
 \draw[latex-latex] (1.5,2.5) -- (7.5,2.5);
 \draw[latex-latex] (1.5,3.5) -- (7.5,3.5);
 \draw[latex-latex] (1.5,4.5) -- (7.5,4.5);
 \draw[latex-latex] (1.5,5.5) -- (7.5,5.5);
 
 \draw[latex-latex] (1.5,1.5) -- (7.5,7.5);
 \draw[latex-latex] (1.5,2.5) -- (7.5,8.5);
 \draw[latex-latex] (1.5,3.5) -- (7.5,9.5);
 \draw[latex-latex] (1.5,4.5) -- (7.5,10.5);
 \draw[latex-latex] (1.5,5.5) -- (7.5,11.5);
 
 \draw[latex-latex] (-1,0) -- (-1,6) node[midway,left]{$p$};
 \draw[latex-latex] (-1,6) -- (-1,12) node[midway,left]{$p$};
 \draw[latex-latex] (0,-1) -- (6,-1) node[midway,below]{$p$};
 \draw[latex-latex] (6,-1) -- (12,-1) node[midway,below]{$p$};
\end{tikzpicture}
}

%% file: figures/foldinglayers.tex
\begin{tikzpicture}[scale=.4]
		
    \begin{scope}[
            xshift=250,yshift=0,yslant=-0.5,xslant=0
            ]
        \fill[white,fill opacity=0.9] (0,0) rectangle (5,5);
        \draw[fill=gray] (0,0) rectangle (7,1) rectangle (6,7) rectangle (0,6) rectangle (1,1);
        \draw[black] (0,0) grid (7,7); 
        \draw[black,very thick] (0,0) rectangle (7,7);
        \draw[-latex,thick] (1.5,1.5) node (begin3) {} -- (5.5,1.5) -- (5.5,2.5) -- (1.5,2.5) -- (1.5,3.5) -- (5.5,3.5) -- (5.5,4.5) -- (1.5,4.5) -- (1.5,5.5) -- (5.5,5.5) node (end3) {};
        \end{scope}    
    \begin{scope}[
            xshift=0,yshift=0,yslant=-0.5,xslant=0
            ]
        \fill[white,fill opacity=0.9] (0,0) rectangle (5,5);
        \draw[fill=gray] (0,0) rectangle (7,1) rectangle (6,7) rectangle (0,6) rectangle (1,1);
        \draw[black] (0,0) grid (7,7); 
        \draw[black,very thick] (0,0) rectangle (7,7);
        \draw[-latex,thick,x=-1cm,xshift=7cm,y=-1cm,yshift=7cm] (1.5,1.5) node (begin2) {} -- (5.5,1.5) -- (5.5,2.5) -- (1.5,2.5) -- (1.5,3.5) -- (5.5,3.5) -- (5.5,4.5) -- (1.5,4.5) -- (1.5,5.5) -- (5.5,5.5) node (end2) {};
         \end{scope}    
    \begin{scope}[
            xshift=-250,yshift=0,yslant=-0.5,xslant=0
            ]
        \fill[white,fill opacity=0.9] (0,0) rectangle (5,5);
        \draw[fill=gray] (0,0) rectangle (7,1) rectangle (6,7) rectangle (0,6) rectangle (1,1);
        \draw[black] (0,0) grid (7,7); 
        \draw[black,very thick] (0,0) rectangle (7,7);
        \draw[-latex,thick] (1.5,1.5) -- (5.5,1.5) -- (5.5,2.5) -- (1.5,2.5) -- (1.5,3.5) -- (5.5,3.5) -- (5.5,4.5) -- (1.5,4.5) -- (1.5,5.5) -- (5.5,5.5) node (end1) {};
     \end{scope}    
     \draw[-latex] (end1) -- (begin2);
      \draw[-latex] (end2) -- (begin3);
   \end{tikzpicture}